\documentclass[a4paper]{article}
\usepackage{microtype}
\usepackage{bm}
\usepackage{fullpage}

\usepackage{amsmath, amssymb}
\usepackage{amsthm}

\title{When the Optimum is also Blind: a New Perspective on Universal Optimization\footnote{An extended abstract of this paper appears in ICALP'17. The second author was partially supported by the ERC Starting Grant NEWNET 279352 and the SNSF Grant APPROXNET 200021\_159697/1.
The fourth author was partially supported by the National Science Centre of Poland Grant UMO-2016/21/N/ST6/00968.}}

\usepackage{color}

\def\DEBUG{true}

\ifdefined\DEBUG
\definecolor{marekgreen}{RGB}{0,185,0}
\definecolor{orange}{RGB}{255,128,0}

  \def\rem#1{{\marginpar{\raggedright\scriptsize #1}}}
  \newcommand{\fabr}[1]{\rem{\textcolor{red}{$\bullet$ #1}}}
  \newcommand{\marr}[1]{\rem{\textcolor{marekgreen}{$\bullet$ #1}}}
  \newcommand{\micr}[1]{\rem{\textcolor{blue}{$\bullet$ #1}}}
  \newcommand{\ster}[1]{\rem{\textcolor{orange}{$\bullet$ #1}}}
\else

  \newcommand{\fabr}[1]{ }
  \newcommand{\marr}[1]{ }
  \newcommand{\micr}[1]{ }
  \newcommand{\ster}[1]{ }

\fi

\author{Marek Adamczyk\thanks{University of Bremen, Germany, m.adamczyk@uni-bremen.de} \and Fabrizio Grandoni\thanks{IDSIA, USI-SUPSI, Lugano, Switzerland, fabrizio@idsia.ch} \and Stefano Leonardi\thanks{Sapienza University of Rome, Italy, leonardi@dis.uniroma1.it} \and Micha\l{} W\l{}odarczyk\thanks{University of Warsaw, Poland, m.wlodarczyk@mimuw.edu.pl}}

\newtheorem{theorem}{Theorem} 
\newtheorem{lemma}[theorem]{Lemma}
\newtheorem{corollary}[theorem]{Corollary}

\global\long\def\mP{\mathbb{P}}
\global\long\def\cC{{\cal C}}

\global\long\def\fs{{\cal S}}
\newcommand{\sm}{\setminus}
\newcommand{\eps}{\varepsilon}
\global\long\def\gpi{{g_\pi}}

\usepackage{verbatim, xspace}

\newcommand{\US}{\textsc{Universal Stochastic}\xspace}

\newcommand{\SC}{\textsc{Set Cover}\xspace}
\newcommand{\VC}{\textsc{Vertex Cover}\xspace}
\newcommand{\EC}{\textsc{Edge Cover}\xspace}
\newcommand{\CSM}{\textsc{Constrained Set Multicover}\xspace}

\newcommand{\NMFL}{\textsc{Non-Metric Facility Location}\xspace}
\newcommand{\MFL}{\textsc{Metric Facility Location}\xspace}

\newcommand{\MC}{\textsc{Multicut}\xspace}
\global\long\def\br#1{\left( #1 \right)}
\global\long\def\brq#1{\left[ #1 \right]}
\global\long\def\eps{\varepsilon}
\global\long\def\ex#1{\mathbb{E}\left[#1\right]}
\global\long\def\exls#1#2{\mathbb{E}_{#1}\left[#2\right]}
\global\long\def\ind#1{\mathbf{1}\left[ #1 \right]}
\global\long\def\chi#1{\chi\brq{#1}}
\global\long\def\pr#1{\mathbb{P}\left[ #1 \right]}
\global\long\def\prls#1#2{\mathbb{P}_{#1}\brq{#2}}
\global\long\def\setst#1#2{\left\{  #1\left|#2\right.\right\}  }
\global\long\def\set#1{\left\{  #1\right\}  }

\global\long\def\bigs{\mbox{\scriptsize big}}
\global\long\def\smalls{\mbox{\scriptsize sml}}
\global\long\def\poly#1{\mbox{poly}\br{#1}}

\begin{document}

\maketitle

\begin{abstract}
\noindent  Consider the following variant of the set cover problem. We are given a universe $U=\set{1,...,n}$ and a collection of subsets $\mathcal{C} = \set{S_1,...,S_m}$ where $S_i \subseteq U$. For every element $u \in U$ we need to find a set $\phi\br{u} \in \mathcal C$ such that $u\in \phi\br u$. Once we construct and fix the mapping $\phi:U \mapsto \mathcal C$ a subset $X \subseteq U$ of the universe is revealed, and we need to cover all elements from $X$ with exactly  $\phi(X):=\bigcup_{u\in X} \phi\br u$. The goal is to find a mapping such that the cover $\phi(X)$ is as cheap as possible.

This is an example of a universal problem where the solution has to be created before the actual instance to deal with is revealed.
Such problems appear naturally in some settings when we need to optimize under uncertainty and it may be actually too expensive to begin finding a good solution once the input starts being revealed.  A rich body of work was devoted to investigate such problems under the regime of worst case analysis, i.e., when we measure how good the solution is by looking at the worst-case ratio: universal solution for a given instance vs optimum solution for the same instance.

As the universal solution is significantly more constrained, it is typical that such a worst-case ratio is actually quite big. One way to give a viewpoint on the problem that would be less vulnerable to such extreme worst-cases is to assume that the instance, for which we will have to create a solution, will be drawn randomly from some probability distribution. In this case one wants to minimize the expected value of the ratio: universal solution vs optimum solution. Here the bounds obtained are indeed smaller than when we compare to the worst-case ratio.

But even in this case we still compare apples to oranges as no universal solution is able to construct the optimum solution for every possible instance.  What if we would compare our approximate universal solution against an optimal universal solution that obeys the same rules as we do?
We show that under this viewpoint, but still in the stochastic variant, we can indeed obtain better bounds than in the expected ratio model. For example, for the set cover problem we obtain $H_n$ approximation which matches the approximation ratio from the classic deterministic offline setup. Moreover, we show this for all possible probability distributions over $U$ that have a polynomially large carrier, while all previous results pertained to a model in which elements were sampled independently. Our result is based on rounding a proper configuration IP that captures the optimal universal solution, and using tools from submodular optimization.

The same basic approach leads to improved approximation algorithms for other related problems, including Vertex Cover, Edge Cover, Directed Steiner Tree, Multicut, and Facility Location.
\end{abstract}

\section{Introduction}

In a typical online problem part of the input is revealed gradually to an algorithm, which has to react to each new piece of the input by making irrevocable choices.
In an online covering problem the online input consists of a sequence of \emph{requests}, which have to be satisfied by the algorithm by buying \emph{items} at minimum total cost.

Some online applications have severe resource constraints, typically in terms of time and/or computational power. Hence even making an online (non-trivial) choice might be too costly. In these settings it makes sense to consider universal algorithms. Roughly speaking, the goal of these algorithms is to pre-compute a reaction to each possible input, so that the online choice can then be made very quickly (say, looking at some pre-computed table). Since the adversary has a lot of power in the universal setting, typically one assumes a stochastic input. In particular, the input is sampled according to some probability distribution $\pi$, which is either given in input or that can be sampled multiple times at polynomial cost per sample (\emph{oracle model}).

The most relevant prior work for this paper is arguably due to Grandoni et al.~\cite{GGLMSS13} (conference version in \cite{GGLMSS08}). The authors consider the universal stochastic version of some classical NP-hard covering problems such as set cover, non-metric facility location, multicut etc. They provide polynomial-time approximation algorithms for those problems in the \emph{independent activation model}, where each request $u$ is independently sampled with some known probability $p_u$. Crucially, in their work the approximation ratio is obtained by comparing the expected cost of the approximate solution with the expected cost of the optimal offline solution (that knows the future sampled input). For example, in the set cover case they present a polynomial-time algorithm that computes a mapping of expected cost at most $O(\log (nm)) \mathbb{E}[OPT_{off}(X)]$, where the expectation is taken over the sampling of $X$ according to $\pi$ and $OPT_{\mbox{\scriptsize off}}(X)$ is the minimum (offline) cost of a set cover of $X$. Here $n$ is size of the universe and $m$ the number of subsets. For $m \gg n$ this ratio becomes $O\br{\frac{\log m}{\log \log m}}$ and is tight. They also consider the universal (non-metric) facility location problem in the independent activation model, and provide a $O\br{\log n}$ approximation (in the above sense), where $n$ is the total number of clients and facilities. We remark that their method seems not to lead to any improved approximation factor in the metric version of the problem. We finally mention their $O(\log^2 n)$ approximation for universal multicut in the independent activation model, where $n$ is the number of nodes in the graph.

\subsection{Our Results and Techniques}
Comparing with the offline optimum as in \cite{GGLMSS13} might be too pessimistic. And often when we need to optimize under uncertainty we cannot really find a better benchmark, flagship example of it would be online problems. However, stochastic two-stage~\cite{IKMM04,RS04} and stochastic adaptive~\cite{MSU99,DGV08,CIKMR09,GN13} problems have proven that one can actually compare an approximate solution with an optimum algorithm that is not omnipotent but obeys the same rules of the model as the approximate one. This inspired us to ask the following question: 
\begin{quote} \emph{Is it also possible in a stochastic universal problem to compare our algorithm with an optimum solution that is restricted by the model in the same way as we are?}
\end{quote}
In this paper we show that we can do this indeed. In this way we manage to obtain tighter approximation ratios --- which of course are compared to a weaker benchmark, but this benchmark itself can be interpreted as more fair and meaningful --- and it also allows us to approach more general problems. 

\subsubsection{Universal Stochastic Set Cover}

We shall describe carefully the \US \SC problem in this section so that we will fully present the model. 
For the remaining problems their full statements will appear in appropriate sections.

In the \US \SC problem we are given in input a universe $U = \set{1,2,...,n}$, and a collection $\cC\subseteq 2^U$ of $m$ subsets $S\subseteq U$, each one with an associated cost $c(S)$. We need to a priori map each element $u\in U$ into some set $\phi(u)\in \cC$.
Then a subset $X\subseteq U$ is sampled according to some probability distribution $\pi$ (whose features are discussed later), and we have to buy the sets $\phi(X)=\bigcup_{u\in X}\phi(u)$ as the cover of $X$. Our goal is to minimize the expected value of the total cost, i.e.,~$\exls{X\sim \pi}{c\br{\phi(X)}}=\exls{X\sim \pi}{\sum_{S\in \phi(X)}c(S)}$. 
One of the most important aspects in our model is that we \underline{do not} compare ourselves against the expected value of an optimum offline solution for a given scenario, that is, \underline{not} against $\exls{X\sim \pi}{OPT_{\mbox{\scriptsize off}}\br{X}}$. 
What we compare ourselves with is
$$
\min_{\phi: \forall_{u\in U}u \in \phi\br{u} }\exls{X\sim \pi}{\sum_{S\in \phi(X)}c(S)},
$$
i.e., the expected cost of an \emph{optimal universal mapping}.

The results depend on the properties of the sampling process. Here we will focus on the most common models, which are defined as follows:
\begin{itemize}
\item \textbf{Scenario model:} Here we are given in input all the sets $X_1,...,X_N \subseteq U$ such that $\prls{X\sim \pi}{X = X_i} > 0$ with the associated probability. This model allows for explicit use of all the scenarios in the computations. For the \US \SC we obtain $O(\log n)$-approximation in this case even in the weighted case.
\item \textbf{Oracle model:}  This is the most general model. We have a black-box access to an oracle $\Pi$ from which we can sample a scenario from distribution $\pi$. We assume that taking a sample requires polynomial time. In this model we can find an $O\br{\log n}$-approximation for \US \SC in polynomial time only for the unweighted case; in the weighted case we achieve the same approximation factor in pseudo-polynomial time depending on the largest cost $\max_{S\in \mathcal{C}} c\br{S}$. We can also show that the same cannot be achieved in polynomial time.
\item \textbf{Independent activation model: } In this model we assume that every element $u \in U$ is independently sampled with some given probability $p_u$. This model does not capture correlations of elements, and therefore sometimes it is not fully realistic. Though it cannot be represented by a polynomial number of scenarios, its nice properties allow one to develop good approximation algorithms for several problems. In this setting we are able to approximate \US \SC within a factor $O\br{\log n}$ in polynomial time even in the weighted case.
\end{itemize}
Our results are obtained by defining a proper configuration LP (with an exponential number of variables) that captures the optimal mapping. We are able to solve this LP via the ellipsoid method using a separation oracle. Somehow interestingly, our separation oracle has to solve a submodular minimization problem. Then we can round the fractional solution in a standard way.

\subsubsection{Overview of the results}
The robustness of our framework allows us to address
universal extensions of several covering problems.
After expressing the goal as a true approximation task,
we can adapt tools from the rich theory of approximation algorithms.

Here we give an overview of our results. Detailed statements of the theorems appear in appropriate sections.

\paragraph*{Scenario model}
In this setting, we are able to construct an LP-based polynomial-time
$O(\log n)$-approximation to the universal stochastic version of \SC (Theorem~\ref{thm:sc}),
which generalizes to \NMFL and \CSM .
In fact, the latter algorithm achieves an approximation guarantee of exactly $H_n$.
Different rounding procedure leads to a 2-approximation for \US \textsc{Vertex Cover}.
All these approximation ratios match the best guarantees obtained in the
deterministic world.
What is more, the exact polynomial time algorithm for \textsc{Edge Cover} extends to the scenario model.

\paragraph*{Independent activation model}
In this setting, we are able to obtain several results in flavour of the $O(1)$-approximation for the \textsc{Maybecast} problem by Karger and Minkoff~\cite{KM00}.
We present a 6.33-approximation for \US \textsc{Metric Facility Location}
(Theorem~\ref{thm:fl}) and an $O(\log n)$-approximation
for \US \MC (Theorem~\ref{thm:mc}).
As an intermediate result, we obtain a 4.75-approximation for \US \MC on trees.

\paragraph*{Oracle model}
We can generalize most of our results for the scenario model to this setting, with the restriction that in the weighted case we get a pseudo-polynomial running time. This is discussed in Section \ref{sec:sc:oracle}.

\subsection{Related work}

Other universal-like problems have been addressed in the literature. For instance, in the universal
TSP problem one computes a permutation of the nodes that is then used to visit a given subset of nodes. This problem has been studied both in the worst-case scenario
for the Euclidean plane~\cite{PB89, BG89} and general
metrics~\cite{Jia05, GHR06, HKL06}, as well as in the
average-case~\cite{Jaillet88,BJO90,SS08,GGLS,ST08}. (For the
related problem of universal Steiner tree,
see~\cite{KM00,Jia05,GHR06,GGLS}.)  
Jia et al.~\cite{Jia05} introduced the 
universal set cover and universal facility location problems, and studied them in the
worst-case: they show that the adversary is very powerful in such models, and give
 nearly-matching $\Omega(\sqrt{n})$ and $O(\sqrt{n\log n})$ bounds on the competitive factor. These problems have been later studied by Grandoni et al. in the independent activation model \cite{GGLMSS13}, as already mentioned before. 

A somewhat related topic is \emph{oblivious
routing}~\cite{Rac02, HHR03,
  BKR03} (see, e.g.,~\cite{VB81, Vocking01} for special cases). A tight logarithmic competitive result as well as a
polynomial-time algorithm to compute the best routing is known in
the worst case for undirected graphs~\cite{ACFKR03,Raecke}. For \emph{oblivious
routing on directed graphs} in the worst case the lower bound of
$\Omega(\sqrt{n})$~\cite{ACFKR03} nearly matches upper
bounds in~\cite{HKLR05-dir} but for the average case. The authors of~\cite{HKLR05-randdir} give an $O(\log^2
n)$-competitive oblivious routing algorithm when demands are
chosen randomly from a known demand-distribution; they also use ``demand-dependent'' routings
and show that these are necessary.

Another closely related notion is the one of \emph{online problems}. These problems have a long
history (see, e.g.,~\cite{BEY98,FW98}), and there have been many
attempts to relax the strict worst-case notion of competitive
analysis: see, e.g.,~\cite{DorriLO05,AlbersLeon99,GGLS} and the
references therein. Online problems with stochastic inputs
(either i.i.d.\ draws from some distribution, or inputs arriving
in random order) have been studied, e.g., in the context of
optimization problems~\cite{Mey-OFL01, MMP01, GGLS, AGLS16}, secretary
problems~\cite{Freeman-survey}, mechanism
design~\cite{HajiaghayiKP04}, and matching problems in
Ad-auctions~\cite{MSVV07}.

Alon et al.~\cite{AAABN} gave the first online algorithm for set
cover with a competitive ratio of $O(\log m \log n)$; they used an
elegant primal-dual-style approach that has subsequently found
many applications (e.g.,~\cite{AAABN04,BuchbinderNaor}).
This ratio is the best possible under complexity-theoretic
assumptions~\cite{Feige-Korman}; even unconditionally, no
deterministic online algorithm can do much better than
this~\cite{AAABN}. Online versions of \emph{metric} facility
location are studied in both the worst
case~\cite{Mey-OFL01,fotakis03}, the average case~\cite{GGLS}, as
well as in the stronger \emph{random
  permutation model}~\cite{Mey-OFL01}, where the adversary chooses a set
of clients unknown to the algorithm, and the clients are presented
to us in a random order. It is easy to show that for our problems,
the random permutation model (and hence any model where elements
are drawn from an \emph{unknown} distribution) are as hard as the
worst case.

One can of course consider the (offline) stochastic version of optimization problems. For example, $k$-stage
stochastic set cover is studied in \cite{IKMM04,SwamyS04}, with an improved approximation factor (independent from $k$) later given in~\cite{Srinivasan07}.

The result for the Oracle model are based on the Sample Average Approximation approach, see~\cite{CCP05} for the application most relevant to our work.

As mentioned before, in two-stage stochastic problems~\cite{RS04,IKMM04} and stochastic adaptive problems~\cite{DGV08,MSU99,CIKMR09,GN13} it is possible to compare a given algorithm with an optimum algorithm which is similarly constrained, and this is what shed a light on the possibility of obtaining results in the same spirit for the universal stochastic optimization.

In two recent papers~\cite{EeIJS16,FMSSSSZ14}, the authors looked at universal optimization over scenarios, but compared against the average offline optimum, and not the optimum universal solution as we do.
All the properties of submodular functions used in our work can be found here~\cite{Vondrak:PhD}.

\section{Preliminaries}
Here we give some basic definitions and properties.
Given a universe $U$, we call a function $f:2^U \rightarrow \mathbb{R}$ submodular if $f(A) + f(B) \ge f(A\cap B) + f(A\cup B)$ for each pair of sets $A,B \subseteq U$. Function $f$ is monotone if $f(B)\geq f(A)$ for $A\subseteq B$. When considering a submodular function, we assume that it is implicitly given in the form of an oracle that can be queried on a specific $A\subseteq U$ and returns the value $f(A)$ in constant time.

\begin{theorem}[Iwata et al. \cite{iwata}]\label{th:sub-poly}
There is an algorithm to minimize a given submodular function $f:2^U \rightarrow \mathbb{N}$ in polynomial time in $|U|$ and in the number of bits needed to encode the largest value of $f$.
\end{theorem}

Let us introduce a function $g_{\pi}:2^U \rightarrow \mathbb{R}^+$, $g_{\pi}(A) = \mP[A \cap X \neq \emptyset]$
where $X$ is drawn from the distribution $\pi$. Our framework exploits crucially the fact that $g_{\pi}$ is a submodular function.

\begin{lemma}\label{lem:sub}
Function $g_{\pi}$ is submodular and monotone.
\end{lemma}
\begin{proof}
 Observe that $g_{\pi}(A) = \sum_{X \subseteq U} \pi(X) \cdot \ind{A \cap X \neq \emptyset}$.
Function $A \rightarrow \ind{A \cap X \neq \emptyset}$ is submodular
and a combination of such functions with positive coefficients is submodular. The monotonicity holds trivially by definition.
\end{proof}

\section{Universal Set Cover}
\label{sec:sc}

In this section we present our approximation algorithm for \US\SC. We start by presenting an $O(\log n)$ approximation in the scenario and independent activation model (Section \ref{sec:sc:scenario}). Then we achieve the same approximation factor in the oracle case (Section \ref{sec:sc:oracle}) though in pseudo-polynomial time (polynomial time for the cardinality case). In Section \ref{sec:sc:lb} we argue that pseudo-polynomial time is indeed needed in order to get a sub-polynomial approximation factor (for the weighted case). 
In the appendix we present approximation algorithms for some special cases (Section \ref{sec:vertexCover}) and some generalizations (Sections \ref{sec:nmfl}, \ref{sec:csm}) of \US\SC.

\subsection{The Scenario and Independent Activation Model}
\label{sec:sc:scenario}


We let $n=|U|$ be the number of elements in the universe. Recall that, for $B\subseteq U$, $g_{\pi}(B)=\mP_{X\sim \pi}[B\cap X\neq \emptyset]$. As mentioned before, $g_{\pi}$ is a submodular function over the universe $U$.
For our goals it is sufficient that $g_\pi$ can be evaluated in polynomial time. This clearly holds both in the scenario model and in the independent activation model.

We start by expressing our problem as the following integer program.
\begin{align}
\min\quad & \sum_{S \in \cC}c(S)\sum_{B\subseteq S}y_{B}^{S}\cdot g_{\pi}(B) & \text{(CONF-IP-SC)}\nonumber \\
\mbox{s.t.} \quad & \sum_{B\ni u}\sum_{S\supseteq B}y_{B}^{S}\geq1 & \forall_{u \in U} \nonumber\\
 & y_{B}^{S} \in \{0,1\} & \forall_{S \in \cC} \forall_{B\subseteq S}  \label{cons:integrality}.
\end{align}
Intuitively, $y_{B}^{S}=1$ means that exactly the elements $B$ of $S$ are mapped into $S$, i.e. $B=\{u\in U: \phi(u)=S\}$.

\begin{lemma}
Integer program (CONF-IP-SC) is equivalent to \US\SC.
\end{lemma}
\begin{proof}
It is easy to translate a mapping $\phi: U \rightarrow \cC$ into some feasible solution to (CONF-LP-SC).
All variables are zeros by default and for each $S \in \cC$ such that $\phi^{-1}(S)$ is non-empty,
we set $y_{\phi^{-1}(S)}^S = 1$.
Note that always $\phi^{-1}(S) \subseteq S$.
In that setting the objective value equals the expected cost of the covering.

Let us fix some feasible solution $\{y_B^S\}_{S \in \cC,\,B\subseteq S}$.
We know that for each $u \in U$ there is some pair $(B,S)$
so that $u\in B$ and $y_B^S = 1$ (we will call it a \textit{covering pair}).
As long as there are many covering pairs for some $u$, we replace one of them
with $(B \sm \{u\}, S)$.
The new solution is still feasible and the objective value is no greater as function
$g_{\pi}$ is monotone.
Therefore there exists an optimal solution so that each $u \in U$ admits exactly one
covering pair $(B_u, S_u)$.
We can define $\phi(u) = S_u$ to obtain a covering with expected cost equal to
the value of the objective function.
\end{proof}

We obtain a linear relaxation (CONF-LP-SC) of (CONF-IP-SC) by replacing the integrality constraints \eqref{cons:integrality} with $y^{S}_{B}\geq 0$. (CONF-LP-SC) has an exponential number of variables: in order to solve it we consider its dual, and provide a separation oracle to solve it. Interestingly, our separation oracle uses submodular minimization.
\begin{lemma}\label{lem:LPsolution}
(CONF-LP-SC) can be solved in polynomial time when $g_\pi$ can be evaluated in polynomial time.
\end{lemma}
\begin{proof}
We show how to solve the dual of (CONF-LP-SC), which is as follows:
\begin{align}
\max\quad & \sum_{u\in U}\alpha_{u} &  \text{(DP-SC)} \nonumber\\
\mbox{s.t.}\quad & \sum_{u\in B}\alpha_{u}\leq c(S)\cdot \gpi(B) & \forall_{S\in\cC} \forall_{B\subseteq S} \nonumber\\
 & \alpha_{u}\geq 0 & \forall_{u \in U} \nonumber.
\end{align}
Observe that (DP-SC) has a polynomial number of variables and an exponential number of constraints. In order to solve (DP-SC), it is sufficient to provide a (polynomial-time) separation oracle, i.e. a procedure that, given a tentative solution $\{\alpha_u\}_{u\in U}$, either determines that it is feasible or provides a violated constraint.

This reduces to check, for each given $S\in \cC$, whether there exists $B\subseteq S$ such that
$\sum_{u\in B}\alpha_{u} > c(S)\cdot g_{\pi}(B)$.
In other terms, we wish to determine whether the minimum of function $h_S(B) := c(S)\cdot g_{\pi}(B) - \sum_{u\in B}\alpha_{u}$ is negative. Observe that the value of $h_S(B)$ can be computed in polynomial time for a given $B$. Note also that $h_S$ is submodular: indeed $g_{\pi}$ is submodular, costs $c(S)$ are non-negative by assumption, and $-\sum_{u\in B}\alpha_{u}$ is linear (hence submodular). Hence we can minimize $h_S$ over $B\subseteq S$ in polynomial time via Theorem \ref{th:sub-poly}.\footnote{
In order to solve the configuration LP, there is an alternative to finding a separation oracle for the dual. We can transform the configuration LP into a  optimization program where we need to minimize a sum of Lovasz's extensions~\cite{Vondrak:PhD}, which are convex functions, over a convex region. This approach would be possibly more efficient, but we have chosen the one above for a simpler presentation.
}
\end{proof}

Given the optimal solution to (CONF-LP-SC), it is sufficient to round it with the usual randomized rounding algorithm for set cover.

\begin{lemma}\label{lem:round}
The optimal solution to (CONF-LP-SC) can be rounded to an integer feasible solution while increasing the cost by a factor $O(\log n)$ in expected polynomial time.
\end{lemma}
\begin{proof}
In the optimal solution, for each $S\in \cC$, variables $\{y^S_B\}_{B\subseteq S}$ define a probability distribution. We sample from this distribution (independently for each $S$) for $q=2\ln n$ many times. Let $B^S_1,\ldots,B^S_q$ be the sets sampled for $S$: we let $B^S:=\cup_i B^S_i$ and tentatively map elements of $B^S$ into $S$. In case the same element $u$ belongs to $B^{S'}$ and $B^{S''}$ for $S'\neq S''$, we replace $B^{S'}$ with $B^{S'}\sm \{u\}$ and iterate: this way each element is mapped into exactly one set. The final sets $B^S$ induce our approximate mapping.

We can upper bound the expected cost of the solution by $\sum_{S\in \cC}c(S)\sum_{i=1}^{q}g_{\pi}(B^S_i)$. Indeed, by the subadditivity of $g_{\pi}$ (which is implied by submodularity and non-negativity), one has $g_{\pi}(B^S)\leq \sum_{i=1}^{q}g_{\pi}(B^S_i)$. Furthermore, $g_{\pi}(B^{S'}\sm \{u\})\leq g_{\pi}(B^{S'})$ since $g_{\pi}$ is monotone. Trivially
\begin{align*}
\ex{\sum_{S\in \cC}c(S)\sum_{i=1}^{q}g_{\pi}(B^S_i)}&=2\ln n \cdot \ex{ \sum_{S\in \cC}c(S)g_{\pi}(B^S_1)} \\
&= 2\ln n \cdot \sum_{S\in \cC}c(S)\sum_{B\subseteq S}y^S_B\cdot g_{\pi}(B^S) \leq 2\ln n \cdot OPT.
\end{align*}
And from Markov's inequality $$\pr{\sum_{S\in \cC}c(S)\sum_{i=1}^{q}g_{\pi}(B^S_i) > 4\ln n \cdot OPT} < \frac{1}{2}.$$

The probability that an element $u \in U$ is not covered with a single sampling over $y^S_B$ is
$$
\prod_{B\ni u}\prod_{S\supseteq B}(1 - y_{B}^{S}) \le \prod_{B\ni u}\prod_{S\supseteq B}e^{-y_{B}^{S}}
= e^{-\sum_{B\ni u}\sum_{S\supseteq B}y_{B}^{S}} \le \frac{1}{e}.
$$
Therefore, by the independence of the sampling and the union bound, the probability that at least one element is not covered is at most $n\cdot \frac{1}{e^{2\ln n}}=\frac{1}{n}$. 

Altogether this gives a Monte-Carlo algorithm. As usual, this can be turned into a Las-Vegas algorithm with expected polynomial running time by repeating the procedure when some element is not covered or the cost of the solution is greater than $4\ln n \cdot OPT$.
\end{proof}

The following theorem and corollary are a straight-forward consequence of Lemmas \ref{lem:LPsolution} and \ref{lem:round}.

\begin{theorem}\label{thm:sc}
\US\SC admits a polynomial-time Las Vegas $O(\log n)$ approximation algorithm w.r.t. the optimal universal solution when $g_{\pi}$ can be evaluated in polynomial time.
\end{theorem}
\begin{corollary}\label{cor:sc}
\US\SC admits a polynomial-time Las Vegas $O(\log n)$ approximation algorithm w.r.t. the optimal universal solution in the scenario model and in the independent activation model.
\end{corollary}


It turns out that not only the randomized rounding technique but also the dual fitting technique can be adapted to the universal stochastic model.
This allows us to improve the above randomized algorithm to a purely deterministic greedy algorithm with slightly improved approximation ratio $H_n$\footnote{Recall that $H_n=1+\frac{1}{2}+\ldots+\frac{1}{n}$ is the $n$-th harmonic number.}. 
The high-level idea is to use a greedy strategy: At each step we select a pair $(B,S)$, $B\subseteq S$, that minimizes $h_S(B):=\frac{c(S)g_\pi(B)}{|B|}$, and remove $B$ from the universe and from all sets. It turns out that finding such pair can also be reduced to submodular minimization (though in a slightly more complicated way). The analysis then follows using standard arguments for greedy set cover. All the details are given in Section~\ref{sec:csm}, where we apply this alternative approach to the more general \US\CSM problem where each element $u$ has to be covered by at least $r(u)$ distinct sets for given positive integer values $\{r(u)\}_{u\in U}$.  

%
%
%

\subsection{The Oracle Model}
\label{sec:sc:oracle}

We next assume that the expected cost of the optimal solution is at least $1$. This is w.l.o.g. (by scaling the minimum set cost to $1$) if $\pi$ does not output empty sets. 

In this case we are not able to compute $g_\pi(B)$ directly for a given $B$, hence we rather try to estimate its value. In more detail, we sample 
$N$ sets $A_{1},...,A_{N}$ from $\pi$, for a sufficiently large $N$ to be fixed later. Then we run the algorithm for the scenario case over the sampled scenarios\footnote{We remark that the latter algorithm is forced to assign all elements to some set, even the elements that happen not to appear in any $A_i$}.

We next analyze the above algorithm, starting from the simpler cardinality case (where all sets have cost $1$). Let us define 
$$\hat{g}\br B=\frac{1}{N}\sum_{i=1}^{N}\ind{A_{i}\cap B\neq\emptyset}.$$ 
Observe that for any feasible integer solution $\{y^S_B\}_{S\in \cC,\,B\subseteq S}$ it holds
\[
\exls{A_{i}\sim\pi}{\sum_{S}c\br S\sum_{B\subseteq S}{y}_{B}^{S}\cdot\ind{A_{i}\cap B\neq\emptyset}}=\sum_{S}c\br S\sum_{B\subseteq S}{y}_{B}^{S}\cdot g_\pi(B).
\]



%

We want to keep track of deviation of such random variables.
We will exploit Chernoff inequality, which states that for i.i.d. variables $X_1,\dots,X_N$
over $[0,1]$ we can estimate
\begin{eqnarray*}
\pr{\frac{1}{N}\br{X_{1}+\dots+X_{N}} > \br{1+\eps}\ex X}\leq\exp\br{-\frac{\eps^{2}}{3}N\cdot\ex X}, \\
\pr{\frac{1}{N}\br{X_{1}+\dots+X_{N}} < \br{1-\eps}\ex X}\leq\exp\br{-\frac{\eps^{2}}{2}N\cdot\ex X}.
\end{eqnarray*}

The value of the objective function for a fixed solution $\{y^S_B\}$ is a random variable over $[1,n]$
so after scaling it to interval $[0,1]$ the expected value is at least $\frac{1}{n}$.
This gives the following bounds
\begin{align*}
\pr{\sum_{S}c\br S\sum_{B\subseteq S}{y}_{B}^{S}\cdot\hat{g}\br B>\br{1+\eps}\sum_{S}c\br S\sum_{B\subseteq S}y_{B}^{S}\cdot \gpi \br B} & \leq\exp\br{-\frac{\eps^{2}}{3}\cdot\frac{N}{n}}
\end{align*}
\begin{align*}
\pr{\sum_{S}c\br S\sum_{B\subseteq S}{y}_{B}^{S}\cdot\hat{g}\br B<\br{1-\eps}\sum_{S}c\br S\sum_{B\subseteq S}y_{B}^{S}\cdot \gpi \br B} & \leq\exp\br{-\frac{\eps^{2}}{2}\cdot\frac{N}{n}}
\end{align*}


If we choose $N \ge \frac{6}{\eps^2}\cdot n(n\log{m} + 2\log{n})$, then
the probability that any of these two bounds does not hold is at most
$\exp(-n\log{m} - 2\log{n})$.
Observe now that the number of possible assignments in \US\SC is $m^n = \exp(n\log{m})$
(we omit solutions in which some element is covered many times).
By the union bound we can estimate that the probability of any bound for any solution
being exceeded is at most $\frac{1}{n^2}$.
Let us name this event as $\mathcal{E}_0$.

Let $Opt$ be the optimal integral solution for (CONF-IP-SC). Let also $\hat{Opt}$ be the optimal integral solution for the set of sampled scenarios, and $\hat{Apx}$ our approximate solution in the same setting. We let $c(y',g')$ denote the cost of a solution $y'$ when in the objective function we replace $\gpi$ with $g'$. So for example $c(\hat{Apx},\gpi)$ is the actual cost of the approximate solution, $c(Opt,\gpi)$ is the optimal cost etc. One has
\begin{align*}
\mathbb{E}[c(\hat{Apx},\gpi)\vert \neg\mathcal{E}_0] & 
\leq \frac{1}{1-\eps}\mathbb{E}[c(\hat{Apx},\hat{g})\vert \neg\mathcal{E}_0]
\leq \frac{O(\log n)}{1-\eps}\mathbb{E}[c(\hat{Opt},\hat{g})\vert \neg\mathcal{E}_0]\\
& \leq \frac{O(\log n)}{1-\eps}\mathbb{E}[c(Opt,\hat{g})\vert \neg\mathcal{E}_0]\leq \frac{O(\log n)(1+\eps)}{1-\eps}\mathbb{E}[c(Opt,\gpi)\vert \neg\mathcal{E}_0]\\
& =\frac{O(\log n)(1+\eps)}{1-\eps}\mathbb{E}[c(Opt,\gpi)].
\end{align*}
Above in the first and fourth inequality we used the event $\neg\mathcal{E}_0$, in the second one the approximation factor of our algorithm for the scenario case, in the third one the optimality of $\hat{Opt}$ w.r.t. $\hat{g}$. Using the fact that $c(\hat{Apx},\gpi)\leq n$ deterministically, one obtains
\begin{align*}
\mathbb{E}[c(\hat{Apx},\gpi)] & =\mathbb{P}[\mathcal{E}_0]\cdot \mathbb{E}[c(\hat{Apx},\gpi)\vert \mathcal{E}_0]+\mathbb{P}[\neg\mathcal{E}_0]\cdot\mathbb{E}[c(\hat{Apx},\gpi)\vert \neg\mathcal{E}_0]\\
& \leq \frac{1}{n^2}n+O(\log n)(1+O(\eps))E[c(Opt,\gpi)] = O(\log n)E[c(Opt,\gpi)].
\end{align*}

\begin{theorem}
Cardinality \US\SC in the oracle model admits a polynomial-time  $O(\log n)$ approximation algorithm.
\end{theorem}


The only difference in the derivation for the weighted case is
that the cost of a solution lies within $[1, Wn]$, where $W$ is the largest set cost (recall that we assume that the minimum set cost is $1$).
Therefore, it suffices to choose $N \ge \frac{6}{\eps^2}\cdot Wn(n\log{m} + 2\log{n} + \log{W})$
to obtain $\pr{\mathcal{E}_0} \le \frac{1}{Wn^2}$ and the rest of the analysis remains the same.
\begin{theorem}\label{thr:sc:oracle:weighted}
\US\SC in the oracle model admits a pseudo-polynomial-time $O(\log n)$ approximation algorithm.
\end{theorem}


We remark that the same reduction from the oracle to the scenario case also works for the other variants of set cover discussed in the appendix. The only difference is in the number $N$ of samples which are required. The simple technical details are left to the reader.



\subsection{A Lower Bound for the Oracle Model}
\label{sec:sc:lb}

The algorithm from Theorem \ref{thr:sc:oracle:weighted} runs in pseudo-polynomial time. We next show that any polynomial-time algorithm for the same setting has approximation ratio $\Omega(\sqrt{n})$. Remarkably, this matches (hence strengthening) the best known lower bound for the deterministic version of \US\SC, even when comparing with the optimal offline solution. 

The basic idea is simple, and can be probably applied to several other problems in the same framework. Intuitively, we consider a lower bound instance for the determistic version of the problem (where the input is adversarially chosen), and \emph{embed} it into the unknown probability distribution with super-polynomially small probability. The rest of the probability mass is assigned to a low-cost dummy subproblem. With large probability any polynomial-time approximation algorithm will not be able to \emph{see} the lower bound instance, and will therefore make \emph{blind} decisions on how to address it. Thus this approximation algorithm in some sense behaves like an algorithm for the deterministic setting. On the other hand the optimal universal solution can be constructed w.r.t. the whole probability distribution, hence achieving the performance of the optimal offline solution in the deterministic setting.   
 
\begin{theorem}\label{thr:sc:lb}
Any (possibly randomized) polynomial-time algorithm for \US\SC in the weighted case has approximation ratio $\Omega(\sqrt{n})$. 
\end{theorem}
\begin{proof}
Let $T=\poly{n,m,\log M}$ be the running time of the considered approximation algorithm, where $n$ is the size of the universe, $m$ the number of sets, and $M$ the largest set weight. Here we assume that $T$ is not a random variable: otherwise a similar argument works with $T$ replaced by, say, $10\cdot \ex{T}$. 

Consider the following input instance. The universe is $U=W\cup \{d\}$. The (dummy) element $d$ is covered only by a singleton set $S_d=\{d\}$ of cost $1$ (in particular, every feasible solution assigns $d$ to $S_d$). Each element $w_i\in W$ is covered by a singleton set $S_i$ of cost $M/\sqrt{n}$. Furthermore there is a set $S_W=W$ of cost $M$ covering precisely $W$. We choose $M$ large enough so that $T/\sqrt{M}=o(1)$.

There are $3$ possible scenarios $X_d=\{d\}$, $X_W=\{W\}$, and $X_w=\{w\}$, where $w$ is an element of $W$ chosen uniformly at random by the adversary. 
Scenario $X_d$ happens with probability $1-1/\sqrt{M}$. Exactly one of the scenarios $X_W$ and $X_w$ happens with the residual probability $1/\sqrt{M}$ according to the following rule. Let $p_W$ denote the probability that, in the cases when the algorithm samples only scenario $X_d$ in the offline stage, then it assigns at least one half of the elements of $W$ to $S_W$ with probability at least $1/2$. Note that $p_W$ is well defined also for randomized algorithms. If $p_W\geq 1/2$ the adversary chooses $\pr{X_w}=1/\sqrt{M}$, and otherwise $\pr{X_W}=1/\sqrt{M}$. Note that in any case the algorithm will sample a scenario different from $X_d$ with probability at most $T/\sqrt{M}=o(1)$

Suppose first that $p_W\geq 1/2$. In this case with probability at least $1/4-o(1)$ the algorithm assigns $w$ to $S_W$. Therefore the expected cost of the approximate solution is at least 
$$
(1/4-o(1))\cdot \frac{1}{\sqrt{M}}\cdot M=\Omega(\sqrt{M}).
$$
A feasible universal solution is $\phi(d)=S_d$ and $\phi(w_i)=S_i$ for all $w_i\in W$. The expected cost of this solution is at most 
$$
(1-\frac{1}{\sqrt{M}})\cdot 1+\frac{1}{\sqrt{M}}\cdot \frac{M}{\sqrt{n}}=O(\sqrt{M/n}).
$$

Suppose next that $p_W<1/2$. In this case with probability at least $1/2-o(1)$ the algorithm assigns at least one half of the elements of $W$ to singleton sets $S_i$, hence paying in expectation at least
$$
(1/2-o(1))\cdot \frac{1}{\sqrt{M}}\cdot \frac{n-1}{2}\cdot \frac{M}{\sqrt{n}}=\Omega(\sqrt{Mn}).
$$
A feasible universal solution is $\phi(d)=S_d$ and $\phi(w_i)=S_W$ for all $w_i\in W$. The expected cost of this solution is at most
$$
(1-\frac{1}{\sqrt{M}})\cdot 1+\frac{1}{\sqrt{M}}\cdot M = O(\sqrt{M}).
$$

In both cases the approximate solution is a factor $\Omega(\sqrt{n})$ worse than the optimal universal solution. 
\end{proof}
We remark that, by the above construction, a polynomial dependence on $M$ in the number of samples (hence in the running time) is needed in order to achieve an $o(\sqrt{n})$ approximation ratio.

\section{Metric Facility Location in the Independent Activation Model}
\label{sec:fl}

In the stochastic universal variant of the uncapacitated facility location problem, we are given a set of clients $C$ and a set of facilities $F$. For each client
$c\in C$ and facility $f\in F$, there is a cost
$d(c,f)\geq 0$ paid if $c$ is connected to $f$; furthermore,
there is a cost $o_{f}\geq 0$ associated with opening facility $f\in F$. In the non-metric version of the problem (considered in Section \ref{sec:nmfl}) we let $d$ be arbitrary, while in the metric case (considered in this section) $d$ induces a metric.
In the universal solution we need to assign every client $c\in C$ to a facility
$\phi(c)\in F$. Then a set $X\subseteq C$ is sampled according to some distribution $\pi$ and we need to open all facilities $\phi(X):=\cup_{c\in X}\{\phi(c)\}$ and connect each $c \in X$ to $\phi(c)$. The goal is to minimize the expected total cost of opening the facilities and connecting clients to facilities, i.e.:
\begin{center}$
\exls{X\sim\text{\ensuremath{\pi}}}{\sum_{f\in \phi(X)}o_{f}+\sum_{c\in X}d\br{c,\phi\br c}}.\label{eq:facility-location}
$\end{center}

In this section we present a constant approximation for the independent activation case where each client $c$ is independently chosen into the scenario with probability $p_{c}$. 
\begin{theorem}\label{thm:fl}
\US \MFL admits a deterministic polynomial-time $\frac{4e}{e-1}$ approximation algorithm w.r.t. the optimal universal solution in the independent activation model.
\end{theorem}

Just by direct modeling of the above formula with the configuration LP we can see that the following is a relaxation of an integer program that solves the problem: 
\begin{align*}
\min & \sum_{f\in F}o_{f}\sum_{B\subseteq C}y_{B}^{f}\cdot  g_{\pi}(B)+\sum_{c\in C}\sum_{f\in F}\gpi(c)\cdot  d\br{c,f}\cdot \sum_{B\subseteq C: c\in B}y_{B}^{f} &   \mbox{(CONF-LP-FL)}\\
\mbox{s.t.} & \ \forall_{c\in C}: \sum_{f\in F}\sum_{B \subseteq C: c\in B}y_{B}^{f}\geq1    \qquad \mbox{and} \qquad \forall_{f\in F}\forall_{B\subseteq C}:  y_{B}^{f}\geq0 .
\end{align*}
Similarly to the set cover case, the interpretation of $y_B^f=1$ is that facility $f$ serves precisely clients $B$. A similar argument also shows that the associated integer program correctly encodes the input problem. Let $OPT_{\text{CONF-LP-FL}}$ be the optimal solution to the above LP. It holds that $OPT_{\text{CONF-LP-FL}}$ is a lower-bound on the expected cost of an optimal universal solution.

We remark that we are able to solve (CONF-LP-FL) also in the scenario case, however in that case we are missing a good rounding procedure: this is left as an interesting open problem.  


Consider the following alternative LP:
\begin{align*}
\min & \sum_{f\in F}o_{f}\cdot\max_{c\in C}\br{x_{c}^{f}}+\sum_{f\in F}o_{f}\cdot\sum_{c\in C}p_{c}\cdot\bar{x}_{c}^{f}+\sum_{c\in C}\sum_{f\in F}p_{c}\cdot d\br{c,f}\cdot\br{x_{c}^{f}+\bar{x}_{c}^{f}} &  & \emph{(LP-FL)}\\
\mbox{\emph{s.t.}} & \ \forall_{c\in C}: \sum_{f\in F}x_{c}^{f}+\bar{x}_{c}^{f}\geq1
\qquad \mbox{\emph{and}} \qquad \forall_{f\in F}\forall_{c\in C}: x_{c}^{f},\bar{x}_{c}^{f}\geq 0,
\end{align*}
and let $OPT_{\text{LP-FL}}$ be the optimum solution cost of (LP-FL).

\begin{lemma}\label{lem:fl:cost1}
It holds that $OPT_{\text{LP-FL}} \leq \frac{e}{e-1}\cdot OPT_{\text{CONF-LP-FL}}$
\end{lemma}
\begin{proof}
We exploit the following simple inequality (see, e.g., \cite{KM00} for a proof):
\begin{equation}
\min\br{1,\sum_{t\in S}p_{t}} \geq \gpi \br S=1-\prod_{t\in S}(1-p_{t})\geq\br{1-\frac{1}{e}}\min\br{1,\sum_{t\in S}p_{t}}. \label{eq:indineq}
\end{equation}
The lower bound in \eqref{eq:indineq} implies that the solution of the following LP is at most $\frac{e}{e-1}$ times bigger than the solution of (CONF-LP-FL):
\begin{align}
\min & \sum_{f\in F}o_{f}\sum_{B\subseteq C}y_{B}^{f}\cdot\min\br{1,\sum_{c\in B}p_{c}}+\sum_{c\in C}\sum_{f\in F}\gpi(c)\cdot  d\br{c,f}\cdot \sum_{B\subseteq C: c\in B}y_{B}^{f} &  & \mbox{}\label{eq:ind-fl-lp} \\
\mbox{s.t. } &  \sum_{f\in F}\sum_{B \subseteq C: c\in B}y_{B}^{f}\geq1 &\forall_{c\in C}\nonumber\\ 
& y_{B}^{f}\geq0 &\forall_{f\in F}\forall_{B\subseteq C}\nonumber.
\end{align}
Let $Big$ be the collection of sets $B\subseteq C$ such that $\sum_{t\in B}p_{t}>1$, and let $Sml$ be the collection of remaining sets $B\subseteq C$ with $\sum_{t\in B}p_{t}\leq 1$.
Define $x_{c}^{f}$ to be the extent to which $c$ was assigned to
$f$ via sets from $Big$, and $\bar{x}_{c}^{f}$ the extent to which $c$
was assigned to $f$ via sets from $Sml$, i.e.,
$x_{c}^{f}=\sum_{B \in Big: c \in B}y_{B}^{f} \mbox{\quad
and \quad}\bar{x}_{c}^{f}=\sum_{B\in Sml: c\in B}y_{B}^{f}.$ Also, for every client $c$ there is an obvious inequality 
$\sum_{B\in Big}y_{B}^{f}  \geq\sum_{B\in Big: c\in B}y_{B}^{f}=x_{c}^{f}$. 
Now we can lower bound~\eqref{eq:ind-fl-lp}:
\begin{align*}
&\sum_{f\in F}o_{f}\br{\sum_{B \in Big}y_{B}^{f}+\sum_{B\in Sml}y_{B}^{f}\cdot\br{\sum_{j\in B}p_{j}}}+\sum_{c\in C}\sum_{f\in F}p_c\cdot  d\br{c,f}\cdot \sum_{B\subseteq C: c\in B}y_{B}^{f}\nonumber \\
= & \sum_{f\in F}o_{f}\br{\sum_{B\in Big}y_{B}^{f}}+\sum_{f\in F}\sum_{c\in C}o_{f}\cdot p_{c}\br{\sum_{B\in Sml: c \in B}y_{B}^{f}}+\sum_{c\in C}\sum_{f\in F}p_c\cdot  d\br{c,f}\cdot \sum_{B\subseteq C: c\in B}y_{B}^{f} \\ \label{eq:bigsmall}
\geq&\sum_{f\in F}o_{f}\cdot\max_{c\in C}\br{x_{c}^{f}}+\sum_{f\in F}\sum_{c\in C}o_{f}\cdot p_{c}\cdot\bar{x}_{c}^{f}+\sum_{c\in C}\sum_{f\in F}p_{c}\cdot d\br{c,f}\cdot\br{x_{c}^{f}+\bar{x}_{c}^{f}}.
\end{align*}
Observe that $\sum_{f\in F}\sum_{B \subseteq C: c\in B}y_{B}^{f}=\sum_{f\in F}x_{c}^{f}+\bar{x}_{c}^{f}\geq 1$ for all $c\in C$. The claim follows.
\end{proof}

Next lemma shows how to round a solution to (LP-FL) into a feasible solution for the original problem.
\begin{lemma}\label{lem:fl:cost2}
There is a polynomial-time deterministic algorithm that computes a solution of cost at most $APX_{\text{LP-FL}}\leq 4 OPT_{\text{LP-FL}}$ with respect to the objective function of (LP-FL).
\end{lemma}
\begin{proof}
(LP-FL) has a polynomial number of variables and constraints, and
so it can be solved in polynomial time: denote by $\br{x_{c}^{f},\bar{x}_{c}^{f}}_{c\in C,f\in F}$ the corresponding optimal solution. We split the clients into two
groups:
\[
C_{\bigs}:=\setst{c\in C}{\sum_{f\in F}x_{c}^{f}\geq\frac{3}{4}}\mbox{ and }C_{\smalls}=\setst{c\in C}{\sum_{f\in F}\bar{x}_{c}^{f}>\frac{1}{4}}.
\] 
We assign the two groups to facilities separately. Consider first clients $C_{\bigs}$. From the definition we get that $\br{\frac{4}{3}x_{c}^{f}}_{c\in C_{\bigs},f\in F}$
is a feasible solution to the following LP:
\begin{align*}
\min\  & \sum_{f\in F}o_{f}\cdot\max_{c\in C_{\bigs}}\br{z_{c}^{f}}+\sum_{c\in C}\sum_{f\in F}p_{c}\cdot d\br{c,f}\cdot z_{c}^{f} &   \\
\mbox{s.t. } & \ \forall_{c\in C_{\bigs}}: \sum_{f\in F}z_{c}^{f}\geq1 \qquad \mbox{and}
 \qquad \forall_{f\in F} \forall_{c\in C_{\bigs}}: z_{c}^{f}\geq0  .
\end{align*}
The corresponding integer program can be interpreted as a variant of standard uncapacitated metric facility location where the underlying metric $d$ is distorted by a factor $p_c$ that depends on client $c$. A folklore result is a primal-dual $3$-approximation algorithm for this problem (this is for example given as an exercise in Vazirani's book \cite{vazirani}). The idea is to modify the classical primal-dual $3$-approximation algorithm so that the dual variable associated to client $c$ grow at speed $p_c$ rather than at uniform speed. We let $\{z_{c}^{f}\}_{c\in C_{big},f\in F}$ be the solution returned by the above rounding algorithm. Thus we have
\begin{align}
& \sum_{f\in F}o_{f}\cdot\max_{c\in C_{\bigs}}\br{z_{c}^{f}}+ \sum_{c\in C_{\bigs}}\sum_{f\in F}p_{c}\cdot d\br{c,f}\cdot z_{c}^{f} \nonumber \\
   \leq  3\cdot&\sum_{f\in F}o_{f}\cdot\max_{c\in C_{\bigs}}\br{\frac{4}{3}x_{c}^{f}}+3\cdot \sum_{c\in C_{\bigs}}\sum_{f\in F}p_{c}\cdot d\br{c,f}\cdot \frac{4}{3}x_{c}^{f} \nonumber \\ 
      \leq  4\cdot&\sum_{f\in F}o_{f}\cdot\max_{c\in C}\br{x_{c}^{f}}+4\cdot \sum_{c\in C}\sum_{f\in F}p_{c}\cdot d\br{c,f}\cdot x_{c}^{f}. \label{eq:boundbiglp}
\end{align}

We next consider clients $C_{\smalls}$. Observe that $\br{4\cdot \bar{x}_{c}^{f}}_{c\in C_{\smalls},f\in F}$ is a feasible solution to the following LP: 
\begin{align*}
\min\ & \sum_{c\in C_{\smalls}}\sum_{f\in F}p_{c}\br{o_{f}+d\br{c,f}}\cdot z_{c}^{f} &   \\ 
\mbox{s.t.} & \ \forall_{c\in C_{\smalls}}:  \sum_{f\in F}z_{c}^{f}\geq1 \qquad \mbox{and} \qquad \forall_{f\in F}\forall_{c\in C_{\smalls}}: z_{c}^{f}\geq0 .
\end{align*}
The above LP has an optimal integral solution: just assign every
client $c\in C_{\smalls}$ to the facility $f$ that minimizes $o_{f}+d(c,f)$.
If $\{z_{c}^{f}\}_{c\in C_{\smalls},f\in F}$
is such an integral solution, then we have that
\begin{eqnarray}
\sum_{c\in C_{\smalls}}\sum_{f\in F}p_{c}\br{o_{f}+d\br{c,f}}\cdot z_{c}^{f}\leq\sum_{c\in C_{\smalls}}\sum_{f\in F}p_{c}\br{o_{f}+d\br{c,f}}\cdot 4\bar{x}_{c}^{f}.\label{eq:boundsmalllp}
\end{eqnarray}
It is then sufficient to return the union of the two mentioned solutions, where $z^f_c=1$ means that $\phi(c)=f$ (breaking ties arbitrarily in case of multiple assignments of the same client). The cost of the overall solution w.r.t. the objective value of (LP-FL) is at most $4\cdot OPT_{\text{LP-FL}}$.
\end{proof}
Now Theorem \ref{thm:fl} follows easily since the cost $APX_{\text{CONF-LP-FL}}$ of the approximate solution with respect to the original objective function satisfies 
$$
APX_{\text{CONF-LP-FL}}\leq APX_{\text{LP-FL}} \overset{Lem. \ref{lem:fl:cost2}}{\leq} 4\cdot OPT_{\text{LP-FL}} \overset{Lem. \ref{lem:fl:cost1}}{\leq} 4\frac{e}{e-1}\cdot OPT_{\text{CONF-LP-FL}},
$$
where in the first inequality above we used the upper bound in \eqref{eq:indineq}.

\section{Multicut in the Independent Activation Model}
\label{sec:mc}

In the (classical version of the) \MC problem we are given an undirected $n$-node graph $G$
with non-negative costs $c_e$ for all $e \in E(G)$ and
a set of pairs of vertices $(s_1,t_1),\dots,(s_k,t_k)$.
The goal is to erase a subset of edges $F$ of minimum cost so that there is no path connecting $s_c$ with $t_c$ for any $c$.
\MC may be considered a covering problem in which the client $c$ is covered
if $F$ contains some $(s_c,t_c)$-cut.

In the universal stochastic setting we are also given a probability distribution $\pi$ over the subsets of $C = [1,k]$
and the solution is a mapping $\phi: C \rightarrow 2^E$ such that
$\phi(c)$ forms a $(s_c,t_c)$-cut.
The expected cost of the solution induced by a mapping $\phi$  equals 
$\exls{X\sim\text{\ensuremath{\pi}}}{\sum_{e\in \phi(X)}c_e}$, where $\phi(X)=\{e\in E: e\in \phi(c)\text{ for some }c\in X\}$.

We express the problem with a configuration integer program.
Let $\mathcal{P}_c$ denote the family of all paths connecting $s_c$ with $t_c$
and let $y_{B}^{e} = 1$ mean that $B = \{c \in C: e \in \phi(c)\}$.
\begin{align}
\min & \sum_{e\in E}c_e\sum_{B\subseteq C}y_{B}^{e}\cdot  \gpi\br B &  & \mbox{(CONF-IP-MC)} 
\\
\mbox{s.t.} & \sum_{e\in P}\sum_{B\ni c}y_{B}^{e}\geq1 &  & \forall_{c\in C}\,\forall_{P\in \mathcal{P}_c}. \nonumber \\
 & y_{B}^{e}{\in \{0,1\}} &  & \forall_{e\in E}\,\forall_{B\subseteq C}. \nonumber
\end{align}

We next use (CONF-LP-MC) to denote the linear relaxation of (CONF-IP-MC). Likewise for \textsc{Facility location}, we will show that in the independent activation
model we can reduce the universal stochastic setting to the \emph{rent-or-buy} setting,
where each edge $e$ can be either bought for price $c_e$ to serve all the clients
or be rented by client $c$ for price $c_e\cdot p_c$\footnote{In the facility location case we did not use the rent-or-buy interpretation explicitly: we do that here in order to give a different viewpoint.}.
We define variables $x^e$ to indicate that $e$ has been bought
and variables $\bar{x}_{c}^{e}$ to express the event of $e$ being {rented} by $c$.
This transition simplifies the linear program CONF-LP-MC to the following one by sacrificing {an} approximation factor {of} $\frac{e}{e-1}$.
\begin{align*}
\min & \sum_{e\in E}c_{e}\cdot x^e + \sum_{e\in E}c_e\sum_{c\in C}p_{c}\cdot\bar{x}_{c}^{e} &  & (\mbox{LP-MC}) 
\\
\mbox{s.t.} & \sum_{e\in P}(x^{e}+\bar{x}_{c}^{e})\geq 1 &  & \forall_{c\in C}\forall_{P\in \mathcal{P}_c} \\
 & x^{e},\bar{x}_{c}^{e}\geq0 &  & \forall_{e\in E}\forall_{c\in C}.
\end{align*}

The proofs of the following two lemmas and corollary are similar to derivations in Section \ref{sec:fl} so we placed them in Appendix~\ref{app:mc}.

\begin{lemma}\label{lem:mc1}
If $\pi$ is an independent activation distribution, then the optimal value of (LP-MC) is at most $\frac{e}{e-1}$
times larger than the optimal value of (CONF-LP-MC).
\end{lemma}

\begin{lemma}\label{lem:mc2}
If $G$ is a tree, then one can round a fractional solution
to (LP-MC) to an integral one of cost at most 3 times larger.
The procedure runs in polynomial time.
\end{lemma}

\begin{corollary}\label{cor:mct}
\US\MC on trees admits a $\frac{3e}{e-1}$-approximation w.r.t. the optimal universal solution in the independent activation model.
\end{corollary}

In order to solve the problem on general graphs,
we will embed the graph into a tree
that approximately preserves the structure of cuts.
The following construction has been introduced by R{\"a}cke~\cite{Raecke}.
We call a tree $T$ a decomposition tree of $G$ if
\begin{enumerate}
\item there is a bijection between the leaves of $T$ and the vertices of $G$,
\item each edge $e_t$ in $T$ has capacity $c^T_{e_t}$ equal to the weight of {the} cut it induces on $V(G)$ (we call this cut $m_T(e_t)$).
\end{enumerate}
For an edge $e \in E(G)$ and a decomposition tree $T$ we define {the} \emph{relative load} of $e$ as
$$
rload_T(e) = \bigg(\sum_{e_t \in E(T): \atop e \in m_T(e_t)} c^T_{e_t} \bigg)\, /\, c_e.
$$

{The main result in \cite{Raecke} concerns the relation between multicommodity flows in $G$ and $T_i$. As our LP formulation is slightly more sophisticated we need to exploit this result in more detail.} Section 2.1 in~\cite{Raecke} describes how to find (in polynomial time) a convex combination
of decomposition trees $\{\lambda_iT_i\}_{i=1}^q$ for a graph $G$, such that $M:=\max_{e \in E(G)}\left[\sum_{i=1}^q \lambda_irload_{T_i}(e)\right] = O(\log n)$.

\begin{theorem}\label{thm:mc}
\US\MC in the independent activation model admits a polynomial-time $O(\log n)$ expected approximation w.r.t. the optimal universal solution.
\end{theorem}
\begin{proof}
Lemma \ref{lem:mc2} implies that a fractional solution to (LP-MC) over $T_i$
can be rounded to an integer solution with an increase of the cost by most a factor $3$.
Observe that each tree edge on a path between terminals corresponds to some cut between these terminals
so any solution to \US\MC on a decomposition tree induces a solution on the original graph
of at most the same cost.

Let $L_i$ denote the optimal value of (LP-MC) over $T_i$
and $L$ denote the same for $G$.
We are going to show that $\sum_{i=1}^q \lambda_iL_i = O(L\log n)$.
In particular this means that $\min_{i=1}^q L_i = O(L\log n)$.
After rounding the fractional solution on $T_i$ of the smallest value,
we will obtain an integral solution for $G$ of cost $O(L\log n)$,
what entails an $O(\log n)$ approximation.

Let us consider the dual linear program of (LP-MC).
\begin{align}
\max & \sum_{c\in C}\sum_{P\in \mathcal{P}_c}\alpha_P &  & \text{(DP-MC)}\nonumber \\
\mbox{s.t.} & \sum_{c\in C}\sum_{P\in \mathcal{P}_c\atop e \in P}\alpha_P \le c_e &  & \forall_{e\in E} \label{eq:mc-lp6-2} \\
 & \sum_{P\in \mathcal{P}_c\atop e \in P}\alpha_P\le c_e\cdot p_c &  & \forall_{c\in C}\,\forall_{e\in E} \label{eq:mc-lp6-3} \\
 & \alpha_P\geq0 &  & \forall_{c\in C}\,\forall_{P\in \mathcal{P}_c}.\nonumber
\end{align}

Feasible solutions to (DP-MC) are just multicommodity flows
satisfying conditions \eqref{eq:mc-lp6-2}-\eqref{eq:mc-lp6-3}.
Let $(\beta^{T_i})$ be an optimal solution to (DP-MC) on a decomposition tree $T_i$ (of value $L_i$)
and let $(\alpha^{T_i})$ be a flow on $G$ where a unit flow over $e_t$ in $(\beta^{T_i})$ translates into a unit flow over $m_T(e_t)$.
Note that the value of shipped commodities remains the same.
Consider $(\alpha) = \sum_{i=1}^q \lambda_i(\alpha^{T_i})$.
This is a (not necessarily feasible) solution of value $\sum_{i=1}^q \lambda_iL_i$.

As $(\beta^{T_i})$ routes at most $c^{{T_i}}_{e_t}$ flow through an edge $e_t$,
then $(\alpha^{T_i})$ routes at most 
$\sum_{e_t \in E(T_i):\atop e \in m_{{T_i}}(e_t)} c^{{T_i}}_{e_t}$
flow through an edge $e$.
Therefore constraint (\ref{eq:mc-lp6-2}) is exceeded at most $rload_{{T_i}}(e)$ times
in $(\alpha^{T_i})$ and $\sum_{i=1}^q \lambda_irload_{T_i}(e)$ times in $(\alpha)$.
If we consider vectors $(\beta^{T_i}_c)$ given by flow routed between terminals of client $c$,
then we conclude the same for constraint (\ref{eq:mc-lp6-3}).

After scaling $(\alpha)$ down times $M = \max_{e \in E(G)}\left[\sum_{i=1}^q \lambda_irload_{T_i}(e)\right]$ we obtain
a feasible solution to (DP-MC) for $G$.
This means that $L$ is no less
than $\sum_{i=1}^q \lambda_iL_i$ divided by $M$.
As we know from \cite{Raecke} that $M = O(\log n)$, the claim follows.
\end{proof}

\bibliographystyle{plainnat}
\bibliography{setcover.bib}

\appendix

\section{Vertex Cover and Edge Cover}
\label{sec:vertexCover}

In this section we consider the universal stochastic versions of Vertex Cover and Edge Cover, which are both special cases of \US\SC. In particular, \US\VC is the special case induced by letting the edges of an undirected graph be the elements of the universe, and the vertices of the same graph be the sets (covering all the edges incident to them). Note that costs are placed on the nodes. Observe also that each vertex might be assigned only a subset of the edges incident to it.
One obtains the \US\EC problem by interchanging the roles of edges and vertices in the above construction (in particular, each edge is a set that covers its two endpoints). We recall that standard \VC admits a $2$ approximation while standard \EC can be solved exactly in polynomial time. Here we show that the same can be achieved in the universal stochastic setting. 

Let us start with \US\VC. We first need the following lemma. Recall that in a set cover instance the \emph{frequency} $f$ of an element $u$ is the number of sets containing $u$.
\begin{lemma}
Under the assumption that each element has frequency at most $f$,
any feasible solution to linear program (CONF-IP-SC) (see Section~\ref{sec:sc})
can be rounded to an integer solution of cost at most $f$ times larger in deterministic polynomial time.
\end{lemma}
\begin{proof}
Fix a solution $\{y^S_B\}_{S\in \cC,\,B\subseteq S}$.
As for each $u\in U$ we have $\sum_{B: u \in B}\sum_{S: B \subseteq S}y_{B}^{S}\geq 1$,
then there is a set $S_u$ such that $\sum_{B: u \in B,\, B \subseteq S_u}y_{B}^{S_u}\geq \frac{1}{f}$
(if there are many such sets, we pick an arbitrary one).
We want to argue that such an assignment provides an integer soution of desirable cost.
Recall that the cost of the solution equals $\sum_{S \in \cC}c(S)\sum_{B\subseteq S}y_{B}^{S}\cdot g_\pi(B)$.
Let $B_S = \{u \in S: S_u = S\}$.
We need to prove that for each set $S$ it holds
\begin{eqnarray*}
g_\pi(B_S) &\le& f\cdot \sum_{B \subseteq S}y^S_B\cdot g_\pi(B), \text{\quad or equivalently} \\
\sum_{X\subseteq U}\pi(X)\cdot\ind{X \cap B_S \neq \emptyset} &\le& f\cdot \sum_{B \subseteq S}y^S_B\cdot \left(\sum_{X\subseteq U}\pi(X)\cdot\ind{X \cap B \neq \emptyset}\right).
\end{eqnarray*}

We will count the contribution of each set $X\subseteq U$ to the left and the right side of the inequality.
If $X \cap B_S = \emptyset$ then $X$ does not influence the left side and may add something to the right side.
Otherwise it adds exactly $\pi(X)$ to the left side.
In this case, let $u$ be any element in $X \cap B_S$.
The contribution of $X$ to the right side equals $\pi(X)\cdot f\cdot \sum_{B \subseteq S,\, B \cap X \neq \emptyset}y^S_B$
which is at least $\pi(X)\cdot f\cdot \sum_{B \subseteq S,\, u \in B}y^S_B \ge \pi(X)\cdot f \cdot \frac{1}{f} = \pi(X)$.
The claim follows by summing over all sets $X$.
\end{proof}

As the linear program (CONF-IP-SC) can be solved in polynomial time (see Lemma~\ref{lem:LPsolution}),
we obtain an algorithm that matches the approximation ratio for deterministic \textsc{Vertex Cover}.

\begin{theorem}\label{thm:vc}
\US\VC admits a polynomial-time deterministic 2-approximation w.r.t. the optimal universal solution when $g_{\pi}$ can be evaluated in polynomial time.
\end{theorem}
\begin{corollary}\label{cor:vc}
\US\VC admits a polynomial-time deterministic 2-approximation w.r.t. the optimal universal solution in the scenario model and in the independent activation model.
\end{corollary}

It turns out that even some exact algorithms can be translated into the universal stochastic paradigm.

\begin{theorem}\label{thm:ec}
\US \textsc{Edge Cover} can be solved exactly in deterministic polynomial time when $g_{\pi}$ can be evaluated in polynomial time.
\end{theorem}
\begin{proof}
Let us consider an edge $e=uv$ with cost $c_e$. Covering both
$u,v$ with $e$ costs $c(e)\cdot g_\pi(\{u,v\}) = c(e)\cdot\mathbb{P}_{X\sim\pi}(u\in X \lor v\in X)$.
Covering only one vertex, e.g. $u$, with $e$ costs $c(e)\cdot\mathbb{P}_{X\sim\pi}(u\in X).$
We construct a new deterministic instance of \textsc{Edge Cover} where each edge $uv$ gets
its cost multiplied by $g_\pi(\{u,v\})$ and we add loops around $u,v$ with costs as above -- all these costs can be
evaluated in polynomial time.

To get rid of the loops we create additional vertices $a,b$, we connect them with a 0-cost
edge, and we transform every loop $vv$ into an edge $va$ with
the same cost.
Covering $a$ and $b$ is free so taking the $va$ edge
is equivalent to covering $v$.
This is a standard instance of \textsc{Edge Cover} which is in P. 
\end{proof}
\begin{corollary}\label{cor:ec}
\US \textsc{Edge Cover} can be solved exactly in deterministic polynomial time in the scenario model and in the independent activation model.
\end{corollary}

\section{Non-Metric Facility Location}
\label{sec:nmfl}

\begin{theorem}\label{thm:nmfl}
\US\NMFL admits a Las-Vegas polynomial-time $O(\log |C|)$ approximation algorithm w.r.t. the optimal universal solution when $g_\pi$ can be evaluated in polynomial time.
\end{theorem}
\begin{proof}
We consider an integer programming formulation for \US\NMFL that resembles CONF-IP-SC for \US\SC (see Section~\ref{sec:sc}), where
there is an additional additive term in the objective function
responsible for the connection costs.
The corresponding linear relaxation and its dual are as follows.
\begin{align*}
\min & \sum_{f \in \mathcal{F}}c(f)\sum_{B\subseteq C}y_{B}^{f}\cdot \gpi(B) + \sum_{f \in \mathcal{F}}\sum_{u \in C}\br{d(u,f)\sum_{B\ni u}y_{B}^{f}}& \\
\mbox{s.t.} & \sum_{B\ni u}\sum_{f \in \mathcal{F}}y_{B}^{f}\geq1 & \forall_{u \in C}\nonumber\\
 & y_{B}^{f} \ge 0 & \forall_{f \in \mathcal{F}}\forall_{B\subseteq C}\nonumber 
\end{align*}

\begin{align*}
\max & \sum_{u\in C}\alpha_{u} &   \\
\mbox{s.t.} & \sum_{u\in B}\alpha_{u}\leq c(f)\cdot \gpi(B) + \sum_{u\in B}d(u,f)& \forall_{f \in \mathcal{F}}\forall_{B\subseteq C}\nonumber\\
 & \alpha_{u}\geq0 & \forall_{u \in C}\nonumber 
\end{align*}

The dual admits a separating oracle alike in Lemma \ref{lem:LPsolution} with the submodular
function to be minimized $h_f(B) = c(f)\cdot \gpi(B) + \sum_{u\in B}d(u,f) - \sum_{u\in B}\alpha_{u}$.
The rounding procedure is analogous to that from Lemma \ref{lem:round}.
\end{proof}
\begin{corollary}\label{cor:nmfl}
\US\NMFL in the scenario model and in the independent activation model admits a Las-Vegas polynomial-time $O(\log |C|)$ approximation algorithm w.r.t. the optimal universal solution.
\end{corollary}

\section{Constrained Set Multicover}
\label{sec:csm}

The \CSM problem is the generalization of \SC where each element $u$ of the universe must be covered by a given positive integer number $r(u)$ of distinct sets. 
In the universal stochastic setting a solution is a mapping
$\phi: U \rightarrow 2^\fs$ such that $|\phi(u)| = r(u)$.
We can express its expected cost as
$\mathbb{E}_{X \sim\pi}\left[\sum_{S\in \phi(X)} c(S)\right]$, where $\phi(X)=\cup_{u\in X}\phi(u)$.

We will adapt the greedy approach for the deterministic \CSM (see Section 13.2.1 in \cite{vazirani}).
For every pair $(B, S),\, S \in \fs,\, B \subseteq S$ we define its \emph{cost-effectiveness}
\begin{eqnarray*}
\hat{c}(B, S)=\frac{c(S)\cdot g_\pi(B)}{|B|}.
\end{eqnarray*}

We are going to maintain a family of sets $(R_{s})_{S\in\fs}$ initiated
with identity $R_{S}=S$.
The set $R_S$ represents those elements that could be covered by $S$ in the future.
As long as there are some not sufficiently covered
elements we choose a pair $(B, S),\, S \in \fs,\, B \subseteq R_S$
minimizing $\hat{c}(B, S)$.
When pair $(B,S)$ is picked, we update $R_S:=R_S \sm B$.
Moreover, when an element $u$ gets covered $r(u)$ times, we erase it from all
$R_{S}$ sets.

Observe that for each $u\in U$ and $S\in\fs$ there is at most
one pair $(B,S)$ that covers $u$ in the solution.
Therefore there is a family of $r(u)$ distinct sets $S$ with this property
and we return it as $\phi(u)$.
If the solution contains multiple pairs $(B_i,S)$ for a single $S$,
then by subadditivity $\gpi(\bigcup_{i=1}^j B_i) \le \sum_{i=1}^j \gpi(B_i)$
and we can replace these sets with their union.

At first we argue that this routine can be implemented in polynomial time.
The only non-trivial part is minimizing the cost-effectiveness.

\begin{lemma}
If we can evaluate $g_{\pi}$ in polynomial time, then we can find a feasible pair minimizing $\hat{c}(B, S)$ in polynomial time.
\end{lemma}
\begin{proof}
We iterate through all $S\in\fs$ and for each of them we want to minimize
the function $h_S(B) = \frac{c(S)\cdot g_\pi(B)}{|B|}$ defined over $R_S$.
As the nominator is a submodular function,  the problem reduces to Lemma \ref{lem:sub-divided}.
\end{proof}

\begin{lemma}\label{lem:sub-divided}
Let $f$ be a submodular function over universe $U$ with integer values in $[0,M]$ and $f(\emptyset)\ge 0$.
The minimum of $\hat{f}(X)=\frac{f(X)}{|X|}$ over $\emptyset\neq X\subseteq U$  can be
found in polynomial time in $|U|$ and in $\log_2 M$.
\end{lemma}
\begin{proof}
We begin with an observation that if two values of $\hat{f}$ differ, then they differ
by at least $\frac{1}{n^2}$.
To see this, consider $1 \le a_1,a_2 \le M,\, 1 \le b_1,b_2 \le n$, such that $\frac{a_1}{b_1} < \frac{a_2}{b_2}$.
Equivalently we can write $a_1b_2 < a_2b_1$.
As the last equations concerns integers we can observe that $a_1b_2 + \frac{b_1b_2}{n^2} \le a_2b_1$,
what implies $\frac{a_1}{b_1} + \frac{1}{n^2} \le \frac{a_2}{b_2}$.
This proves the observation.

We will take advantage of the binary search.
It suffices to check if there is a non-empty set satisfying $\frac{f(X)}{|X|} < c$ for a constant $c$.
It follows from the observation above, that we require at most $\log_2(Mn^2)$ iterations to converge
and we have analogous bound on the binary length of $c$.
To answer the question for a given $c$, we may equivalently ask whether
the minimum of the function $\tilde{f}_c(X) = f(X) - c|X|$ over $X\subseteq U$ is negative
(note that $\tilde{f}_c(\emptyset) = f(\emptyset) \ge 0$ so it does not influence the answer).
The function $\tilde{f}_c$ is submodular and the encoding length of the function value is $O(\log M +\log n)$.
so the question can be answered in polynomial time due to Theorem \ref{th:sub-poly}.
\end{proof}

\begin{theorem}\label{thm:csm}
\US \CSM admits an $H_n$-approximation w.r.t. the universal optimal solution when $g_{\pi}$ can be evaluated in polynomial time.
\end{theorem} 
\begin{proof}
Consider the following linear relaxation of the problem:
\begin{align}
\min & \sum_{S \in\fs}c(S)\sum_{B\subseteq S}y_{B}^{S}\cdot \gpi(B) & (\text{CONF-LP-CSM})\nonumber\\
\mbox{s.t.} & \sum_{B\subseteq S:u\in B}y_{B}^{S}\ge r(u) & \forall_{u \in U}\nonumber\\
 & \sum_{B\subseteq S}y_{B}^{S}\le1 & \forall_{S\in\fs}\nonumber\\
 & y_{B}^{S}\ge0 & \forall_{S\in\fs,\,B\subseteq S}\nonumber.
\end{align}
Here the variables $y^S_B$ have the usual interpretation. Note that the second constraint is now needed to avoid that a set is used to cover multiple times the same element. The reader might easily check that integral solutions to (CONF-LP-CSM) are in one-to-one correspondence with feasible solutions to the original problem.

Next consider the dual of (CONF-LP-CSM):
\begin{align}
\max & \sum_{e}r(u)w_{e}-\sum_{S}z_{S} & (\text{DP-CSM})\nonumber\\
\mbox{s.t.} & \sum_{e\in B}w_{e}-z_{S}\le c(S)\cdot\gpi(B) & \forall_{S\in\fs,\,B\subseteq S}\nonumber\\
 & w_{u}\ge0,z_{S}\ge0 & \forall_{u\in U}\forall_{S\in \fs}\nonumber
\end{align}

It is convenient to imagine each element $u$ as a set of copies $u(1),\ldots,u(r(u))$ which are covered by distinct sets, where $u(i)$ is the $i$-th copy of $u$ to be covered by the greedy algorithm. Let us define $price(u,i)$ to be the cost-effectiveness of the pair $(B,S)$ that
covered $u(i)$.
For $u \in U,\, S \in \fs$ we define $j_u^{S}$ to be the number
of the copy of $u$ covered by a pair $(B,S)$ (for some $B \subseteq S$) or $r(u)$ if $u$ has not
been covered by any such pair. Recall that there can be at most one pair satisfying this condition.
We also define
\begin{eqnarray*}
\alpha_{u} & = & price(u,r(u))\\
\beta_{S} & = & \sum_{u \in U}price(u,r(u))-price(u,j_{u}^{S})
\end{eqnarray*}

We have $price(u,i)\le price(u,i+1)$ so $price(u,j_{u}^{S}) \le price(u,r(u))$
and $\beta_{S}\ge0$.
Observe that the total cost we pay in the algorithm equals
\begin{eqnarray*}
\sum_{u\in U}\sum_{i=1}^{r(u)}price(u,i)=\sum_{u\in U}r(u)\alpha_{u}-\sum_{S}\beta_{S},
\end{eqnarray*}
that is the objective function of the dual linear program (DP-CSM) for
variables $(\alpha, \beta)$.
In order to show that the obtained solution is a $H_n$-approximation,
we need to prove that $\left(\frac{\alpha}{H_{n}}, \frac{\beta}{H_{n}}\right)$ is a feasible
solution to (DP-CSM),
which is equivalent to proving that $\forall_{S\in\fs,\,B\subseteq S}$ it holds
\begin{eqnarray*}
\sum_{u\in B}price(u,j_{u}^{S}) \le c(S)\cdot \gpi(B)\cdot H_{n}.
\end{eqnarray*}

We fix a pair $S\in\fs,\,B\subseteq S$.
The summand $price(u,j_{u}^{S})$ is the cost-effectiveness of the
set covering $u$ in the moment it got removed from $R_{S}$.
Let us order the elements of $B$ in the order they were removed from $R_{S}$:
$u_{1},u_{2},\dots,u_{k}$.
Observe that $u_{i}$ could be covered at that moment by $(B_i, S)$ where
$B_{i}={u_{i},\dots,u_{k}}$, so
\begin{eqnarray*}
& price(u_i,j_{e_i}^S) & \le \frac{c(S)\cdot\gpi(B_{i})}{|B_{i}|} \le \frac{c(S)\cdot\gpi(B)}{k-i+1}, \\
\sum_{u\in B} & price(u,j_{u}^{S}) & \le c(S)\cdot\gpi(B)\cdot \sum_{i=1}^{|B|} \frac{1}{k-i+1} \le c(S)\cdot\gpi(B)\cdot H_n.
\end{eqnarray*}

This shows that $(\frac{\alpha}{H_{n}},\frac{\beta}{H_{n}})$
is a feasible solution to linear program (DP-CSM).
Therefore the cost of the solution is no greater than $H_n$ times the optimum
of (DP-CSM) which equals the optimum of (CONF-LP-CSM).
The claim follows.
\end{proof}
\begin{corollary}\label{cor:csm}
\US \CSM admits an $H_n$-approximation w.r.t. the universal optimal solution in the scenario model and in the independent activation model.
\end{corollary}

\section{Multicut: proofs of Lemmas \ref{lem:mc1} and \ref{lem:mc2}}
\label{app:mc}

\begin{proof}[Proof of Lemma \ref{lem:mc1}]

We want to transform the following LP
\begin{align}
\min & \sum_{e\in E}c_e\sum_{B\subseteq C}y_{B}^{e}\cdot  \gpi\br B &  & \label{eq:mc-lp8} \\
\mbox{s.t.} & \sum_{e\in P}\sum_{B\ni c}y_{B}^{e}\geq1 &  & \forall_{c\in C}\forall_{P\in \mathcal{P}_c}. \nonumber \\
 & y_{B}^{e}\geq0. &  & \forall_{e\in E}\forall_{B\subseteq C}. \nonumber
\end{align}
\noindent Similarly to Section \ref{sec:fl} let us inject the bound from~\cite{KM00}:
\[
 \gpi\br B=1-\prod_{j\in B}(1-p_{j})\geq\br{1-\frac{1}{e}}\min\br{1,\sum_{j\in B}p_{j}}
\]

\noindent into the linear program (\ref{eq:mc-lp8}).
Here is what we obtain:
\begin{align}
\min & \sum_{e\in E}c_e\sum_{B\subseteq C}y_{B}^{e}\cdot \min\br{1,\sum_{j\in B}p_{j}} &  & \label{eq:mc-lp2} \\
\mbox{s.t.} & \sum_{e\in P}\sum_{B\ni c}y_{B}^{e}\geq1 &  & \forall_{c\in C}\forall_{P\in \mathcal{P}_c}. \nonumber \\
 & y_{B}^{e}\geq0. &  & \forall_{e\in E}\forall_{B\subseteq C}. \nonumber
\end{align}

 Observe that the value of the optimal solution to (\ref{eq:mc-lp2}) is at most $\frac{e}{e-1}$
times larger than the optimum of (\ref{eq:mc-lp8}).
On the other hand every integral solution to (\ref{eq:mc-lp2}) translates
into an integral solution to (\ref{eq:mc-lp8}) of the same cost
because $\min\br{1,\sum_{j\in B}p_{j}} \ge \gpi(B)$.

Let $Big$ be a collection of all the sets $B$ such that $\sum_{j\in B}p_j>1$, and 
let $Sml$ be a collection of all the sets $B$ such that $\sum_{j\in B}p_j\leq 1$.
We can rewrite the objective function of (\ref{eq:mc-lp2}) as 
\begin{align}
 & \sum_{e\in E}c_e\sum_{B\subseteq C}y_{B}^{e}\cdot \min\br{1,\sum_{j\in B}p_{j}} \nonumber \\
= & \sum_{e\in E}c_e\br{\sum_{B\in Big}y_{B}^{e} + \sum_{B\in Sml}y_{B}^{e}\br{\sum_{j\in B}p_{j}}} \nonumber \\
= & \sum_{e\in E}c_e\br{\sum_{B\in Big}y_{B}^{e}} + \sum_{e\in E}c_e\sum_{c\in C}p_c\br{\sum_{B\in Sml: c\in B}y_{B}^{e}} \label{eq:mc-lp3}
\end{align}
\noindent Let $x_{c}^{e}=\sum_{B\in Big: c\in B}y_{B}^{e}$
and $\bar{x}_{c}^{e}=\sum_{B\in Sml: c\in B}y_{B}^{e}$. Also, we have
an obvious inequality:
\begin{align*}
\forall c\in C:\qquad\sum_{B\in Big}y_{B}^{e} & \geq\sum_{B\in Big: c\in B}y_{B}^{e}=x_{c}^{e},
\end{align*}
and so~(\ref{eq:mc-lp3}) is greater than
\begin{align*}
\sum_{e\in E}c_e\cdot\max_{c\in C}\br{x_{c}^{e}} + \sum_{e\in E}c_e\sum_{c\in C}p_c\cdot\bar{x}_{c}^{e}.
\end{align*}

\noindent The condition $\sum_{e\in P}\sum_{B\ni c}y_{B}^{e}\geq 1$ translates
into $\sum_{e\in P}(x_c^{e}+\bar{x}_{c}^{e})\geq 1$.
Observe that for each edge $e$ we can replace all variables $x_c^e$
with $x^e = \max_{c\in C}\br{x_{c}^{e}}$: the objective does not change
and the conditions remain satisfied.
Finally, we obtain a new LP with the value of the optimal solution at most $\frac{e}{e-1}$
times the optimum of (\ref{eq:mc-lp8}), which is what we have claimed.
\begin{align}
\min & \sum_{e\in E}c_{e}\cdot x^e + \sum_{e\in E}c_e\sum_{c\in B}p_c\cdot\bar{x}_{c}^{e} &  & \label{eq:mc-lp5} \\
\mbox{s.t.} & \sum_{e\in P}(x^{e}+\bar{x}_{c}^{e})\geq 1 &  & \forall_{c\in C}\forall_{P\in \mathcal{P}_c} \nonumber \\
 & x^{e},\bar{x}_{c}^{e}\geq0 &  & \forall_{e\in E}\forall_{c\in C}.\nonumber
\end{align}
\end{proof}
\begin{proof}[Proof of Lemma \ref{lem:mc2}]

Let us consider a restricted version of \US\MC with an additional assumption
that the graph $G$ is a tree.
In this case each family $\mathcal{P}_c$ consists of only one path -- let us denote it by $P_c$.
We will take advantage of that fact in order to round the linear program (\ref{eq:mc-lp5}).
At first, note that it has a polynomial number of variables and constraints
so it can be solved in polynomial time.
Define $\br{x^e,\bar{x}_{c}^{e}}_{c\in C,e\in E}$ to be
the optimal solution and denote its value as $OPT_{LP}$. We split $C$ into two
groups:
\[
C_{\bigs}:=\setst{c\in C}{\sum_{e\in P_c}x^e \geq\frac{2}{3}}\mbox{ and }C_{\smalls}=\setst{c\in C}{\sum_{e \in P_c}\bar{x}_{c}^{e}>\frac{1}{3}}.
\]

Observe that $\br{\frac{3}{2}\cdot x^e}_{e\in E}$ is a feasible solution to the linear relaxation
of the standard \MC problem on trees with the cut constraints given by $C_{\bigs}$.
This problem admits a 2-approximation with respect to the LP optimum~\cite{garg1997primal}.
Therefore, we can construct an integral solution $\br{z^e}_{e\in E}$ for (\ref{eq:mc-lp5}) satisfying constraints from $C_{\bigs}$
of cost not exceeding $2\cdot\frac{3}{2}\sum_{e\in E}c_{e}\cdot x^e$.

On the other hand $\br{3\cdot\bar{x}_{c}^{e}}_{c\in C_{\smalls},e\in E}$ forms
a feasible solution to the following LP.
\begin{align*}
\min & \sum_{e\in E}c_e\sum_{c\in C}p_c\cdot\bar{x}_{c}^{e} &  & \\
\mbox{s.t.} & \sum_{e\in P_c}\bar{x}_{c}^{e}\geq1 &  & \forall_{c\in C_{\smalls}}\\
 & \bar{x}_{c}^{e}\geq0 &  & \forall_{e\in E}\,\forall_{c\in C_{\smalls}}.
\end{align*}

 It is easy to see that the optimal solution is integral and can be obtained
by assigning each client $c$ to the cheapest edge along the path $P_c$
--- denote this solution as $\br{\bar{z}_{c}^{e}}_{c\in C,\,e\in E}$ where
$\bar{z}_{c}^{e} = 0$ for $c \in C_{\bigs}$.
The total cost of the constructed solution is
\begin{equation*}
 \sum_{e\in E}c_{e}\cdot z^e + \sum_{e\in E}c_e\sum_{c\in C}p_c\cdot\bar{z}_{c}^{e} \le 2\cdot\sum_{e\in E}c_{e}\cdot \br{\frac{3}{2}x^{e}} + \sum_{e\in E}c_e\sum_{c\in C}p_c\cdot\br{3\cdot\bar{x}_{c}^{e}} = 3\cdot OPT_{LP}.
\end{equation*}
\end{proof}

Corollary \ref{cor:mct} directly follows from the above two lemmas.

\end{document}